\documentclass[proceedings]{stacs}
\stacsheading{year}{numbers}{city}
\usepackage{amsmath,amssymb}
\usepackage{enumerate}
\usepackage{ dsfont }

\usepackage{bm}
\usepackage{latexsym}
\usepackage[english]{babel}

\def \O {\mathcal{O}}
\def \A {\mathcal{A}}

\def \pphi {\delta}
\def \Pphi {\Delta}

\def \cg {[ \negthinspace [ }
\def \cd {] \negthinspace ] }
\def\iver#1{\cg {#1}\cd }

\newcommand{\be}{\begin{equation}}
\newcommand{\ee}{\end{equation}}

\newcommand{\ef}[1]{\, #1}

\newcommand{\eee}[1]{\mathbb{E}{#1}}
\newcommand{\ppp}[1]{\mathbb{P}{#1}}

\newcommand{\proba}{\mathrm{prob}}

\def \cg {[ \negthinspace [ }
\def \cd {] \negthinspace ] }
\def\iver#1{\cg {#1}\cd }

\newcommand{\w}{\omega}
\renewcommand{\a}{\alpha}

\newcommand{\cT}{\mathcal{T}}
\newcommand{\cE}{\mathcal{E}}
\newcommand{\cA}{\mathcal{A}}

\newcommand{\cM}{\mathcal{M}}
\newcommand{\bbm}{{\bm m}}

\newcommand{\poiss}{\mathrm{Poiss}}

\newcommand{\dx}[1] {\!\mathrm{d}{#1}\,}

\renewcommand{\emptyset}{\varnothing}
\newcommand{\smfrac}[2]{\genfrac{}{}{0.25pt}{1}{#1}{#2}}
\newcommand{\atopp}[2]{\genfrac{}{}{0pt}{}{#1}{#2}}

\newcommand{\eval}[1]{\left\langle {#1} \right\rangle}

\newcommand{\bc}{{\bm c}}

\title{Asymptotic enumeration of Minimal Automata}
\author[ref1]{F.\ Bassino}{Frederique Bassino}
\author[ref1]{J.\ David}{Julien David}
\author[ref2]{A.\ Sportiello}{Andrea Sportiello}
\address[ref1]{LIPN,
Universit\'e Paris 13, and CNRS. 
99, av.\ J.-B.\ Cl\'ement, 93430 Villetaneuse, France}
\email{Frederique.Bassino@lipn.univ-paris13.fr}
\email{Julien.David@lipn.univ-paris13.fr}
\address[ref2]{Universit\`a degli Studi di Milano, Dip.~di
  Fisica, and INFN.
Via G.~Celoria 16, 20133 Milano, Italy}
\email{Andrea.Sportiello@mi.infn.it}

\begin{document}

\begin{abstract} 
We determine the asymptotic proportion of minimal automata, within
$n$-state accessible deterministic complete automata over a $k$-letter
alphabet, with the uniform distribution over the possible transition
structures, and a binomial distribution over terminal states, with
arbitrary parameter $b$. It turns out that a fraction $\sim 1-C(k,b)
\, n^{-k+2}$ of automata is minimal, with $C(k,b)$ a function,
explicitly determined, involving the solution of a transcendental
equation.
\end{abstract}

\maketitle

\section{Introduction}

To any regular language, one can associate in a unique way its minimal
automaton, i.e.\ the only accessible complete deterministic automaton
recognizing the language, with minimal number of states.  Therefore
the \emph{space complexity} of a regular language can be seen as the
number of states of its minimal automaton.  The worst-case complexity
of algorithms dealing with finite automata is most of times known
\cite{YZS}.  But the average-case analysis of algorithms requires
weighted sums on the set of possible realizations, and in particular
the enumeration of the objects that are handled \cite{flaj}.
Therefore a precise enumeration is often required for the algorithmic
study of regular languages.

The enumeration of finite automata according to various criteria (with or
without initial state \cite{kor78}, non-isomorphic \cite{Har65}, up to
permutation of the labels of the edges \cite{Har65}, with a strongly
connected underlying graph \cite{lis71,kor78,rob85,kor86}, acyclic
\cite{lis03},\ldots)
has been investigated since the fifties.

In \cite{kor78} Korshunov determines the asymptotic estimate of the
number of accessible complete and deterministic $n$-state automata
over a finite alphabet.  His derivation, and even the formulation of
the result, are quite complicated.  In~\cite{BN07} a reformulation of
Korshunov's result leads to an estimate of the number of such automata
involving the Stirling number of the second kind.  On the other side,
in \cite{Leb} a different simplification of the involved expressions
is achieved, by highlighting the role of the Lagrange Inversion
Formula in the analysis.


A natural question is to ask which is the fraction of minimal
automata, among accessible complete and deterministic automata of a
given size $n$ and alphabet cardinality $k$.  Nicaud \cite{Nicaud99}
shows that, asymptotically, half of the complete deterministic
accessible automata over a unary alphabet are minimal, thus solving
the question for $k=1$.  Using {\sf REGAL}, a C++-library for the
random generation of automata, the proportion of minimal automata
amongst complete deterministic accessible ones experimentally seems to
be $85,32\%$ for a $2$-letter alphabet and more than $99,99\%$. for a
larger alphabet~\cite{BDN07}.

In this paper we solve this question for a generic integer $k \geq
2$. At a slightly higher level of generality, we give a precise
estimation of the asymptotic proportion of minimal automata, within
$n$-state accessible deterministic complete automata over a $k$-letter
alphabet, for the uniform distribution over the possible transition
structures, and a binomial distribution over terminal states, with
arbitrary parameter $0<b<1$ (the uniform case corresponding to
$b=\frac{1}{2}$). Our theoretical results are in agreement with the
experimental ones.

The paper is organized as follows.  In Section~\ref{sec.statem} we
recall some basic notions of automata theory, and we set a list a
notations that will be used in the remainder of the paper.  Then, we
state our main theorem, and give a short and simple heuristic
argument. In Section \ref{sec.proofstruct} we give a detailed
description of the proof structure, and its subdivision into separate
lemmas. In Section \ref{sec.proofs} we prove in detail the most
difficult lemmas, and give indications for those that are
provable through standard methods. Finally, in Section \ref{sec:other}
we discuss some of the implications of our result.

\section{Statement of the result}
\label{sec.statem}

\noindent
For a given set $E$, $|E|$ denotes the cardinal of $E$. The symbol
$[n]$ denotes the canonical $n$-element set $\{1,2,\ldots,n\}$.
%
%
Let $\cE$ be a Boolean condition, the Iverson bracket
$\iver{\cE}$ is equal to $1$ if $\cE = \textrm{true}$
and $0$ otherwise.
We use $\eee(X)$ to denote the expectation of the quantifier $X$, and
$\ppp(\mathcal{E}) = \eee( \iver{\cE})$ for the probability of the event~$\mathcal{E}$.
For $\{\cE_i\}$ a collection of events, we define a shortcut for the
first moment
\be
{\bm m}(\{\cE_i\})
:=
\sum_i \, \ppp(\cE_i)
=
\eee{} \Big(
\sum_i \; \iver{\cE_i}
\Big)
\ef.
\end{equation}
If $p(c)$ is the probability that exactly $c$ events occur, we have 
${\bm m}(\{\cE_i\})=\sum_c c\, p(c) \geq \sum_{c \geq 1} p(c) = 1-p(0)$,
i.e.\ $p(0) \geq 1 - {\bm m}(\{\cE_i\})$. This elementary inequality, known as
\emph{first-moment bound}, is used repeatedly in the following.

A {\em finite deterministic automaton} $A$ is a quintuple
$A=(\Sigma,Q,\pphi ,q_0,\cT)$ where $Q$ is a finite set of {\em states},
$\Sigma$ is a finite set of {\em letters} called {\em alphabet}, the
{\em transition function} $\pphi$ is a mapping from $Q\times \Sigma$
to $Q$, $q_0\in Q$ is the {\em initial state} and $\cT\subseteq Q$ is
the set of \emph{terminal} (or \emph{final}) states. With abuse of
notations, we identify $\cT(i)\equiv \iver{i \in \cT}$.

An automaton is {\em complete} when its transition function is total.
The transition function can be extended by morphism to all words of
$\Sigma^*$: $\pphi(p,\varepsilon) =p$ for any $p\in Q$ and for any
$u,v\in \Sigma^*$, $\pphi(p, (uv)) = \pphi(\pphi(p,u),v)$. A word
$u\in \Sigma^*$ is \emph{recognized} by an automaton when 
$\pphi(q_0,u)\in \cT$. The \emph{language} recognized by an automaton is the set
of words that it recognizes. An automaton is {\em accessible} when for
any state $p\in Q$, there exists a word $u \in \Sigma^*$ such that
$\pphi(q_0,u) = p$.

We say that two states $p$, $q$ are \emph{Myhill-Nerode-equivalent}
(or just \emph{equivalent}), and write $p \sim q$, if, for all finite
words $u$, $\cT(\pphi(p,u)) = \cT(\pphi(q,u))$ \cite{Ner58}. This
property is easily seen to be an equivalence relation. An automaton is
said to be \emph{minimal} if all the equivalence classes are atomic,
i.e.\ $p \not\sim q$ for all $p \neq q$. Otherwise, the minimal
automaton $A'$ recognising the same language as $A$ has set of states
$Q'$ corresponding to the set of equivalence classes of $A$. This
automaton can be determined through a fast and simple algorithm, due
to Hopcroft and Ullman. For this and other results on automata see
e.g.\ \cite{Hop-Ull, js}.

At the aim of enumeration, the actual labeling of states in $Q$ and
letters in $\Sigma$ is inessential, and we can canonically assume that
$Q=[n]$, $\Sigma=[k]$, and $q_0=1$. In this case, when there is no
ambiguity on the values of $n$ and $k$, we will associate an automaton
$A$ to a pair $(\pphi,\cT)$, of transition function, and set of
terminal states.  The set of complete deterministic accessible
automata with $n$ states over a $k$-letter alphabet is noted
$\A_{n,k}$.

We will determine statistical averages of quantities associated to
automata $A \in \cA_{n,k}$. This requires the definition of a measure
$\mu(A)$ over $\cA_{n,k}$. The simplest and more natural case is just
the uniform measure. We generalise this measure by introducing a
continuous parameter.  For $S$ a finite set, the
\emph{multi-dimensional Bernoulli distribution} of parameter $b$ over
subsets $S' \subseteq S$ is defined as $\mu_{b}(S') = b^{|S'|}
(1-b)^{|S|-|S'|}$. The distribution associated to the quantifier
$|S'|$ is thus the binomial distribution.  We will consider the family
of measures $\mu_{b}^{(n,k)}(A)=
\mu_{\rm unif}^{(n,k)}(\pphi) \mu_{b}^{(n)}(\cT)$, with $\mu_{\rm
  unif}^{(n,k)}(\pphi)$ the uniform measure over the transition
structures of appropriate size, and $\mu_{b}^{(n)}(\cT)$ the Bernoulli
measure of parameter $b$ over $Q \equiv [n]$.  The uniform measure
over all accessible deterministic complete automata is recovered
setting $b=\frac{1}{2}$.  Superscripts
will be omitted
when clear.


The result we aim to prove in this paper is
\begin{theorem}
\label{theo.fullshort}
In the set $\cA_{n,k}$,
with the uniform measure, the
asymptotic fraction of minimal automata is 
\be
\exp \big( -\smfrac{1}{2} c_k n^{-k+2} \big)
\ef,
\end{equation}
with
\begin{align}
c_k &= 
\smfrac{1}{2}\, {\omega_k}^k
\ef;
&
-k\, \w_k &= \ln (1-\w_k)
\ef.
\label{eq.67476548}
\end{align}
More generally, for any $0<b<1$, with measure
$\mu_{b}^{(n,k)}(A)$, the asymptotic fraction is
\be
\exp \big( -(1-2b(1-b)) c_k n^{-k+2} \big)
\ef.
\end{equation}
\end{theorem}
We singled out the constant $\w_k$, instead of only $c_k$, because the
former appears repeatedly, in the evaluation of several statistical
properties of random automata.  Solving (\ref{eq.67476548}), it can
be written in terms of (a branch of) the Lambert $W$-function, as
$\omega_k = 1+ \frac{1}{k} W(-k e^{-k})$, however the implicit
definition (\ref{eq.67476548}) is more of practical use. See
Table \ref{tab.kcw} for a numerical table of values.

\begin{table}
\[
\begin{array}[tb]{r|ccccc}
k & 2 & 3 & 4 & 5 & 6 \\
\hline
\w_k & 0.796812 & 0.940480 & 0.980173 & 0.993023 & 0.997484 \\
c_k  & 0.317455 & 0.415928 & 0.461509 & 0.482799 & 0.492498
\end{array}
\]
\caption{\label{tab.kcw}The constants involved in the statement of 
Theorem \ref{theo.fullshort}, for the first values of~$k$.}
\end{table}

When it is understood that $|\Sigma|=k$, a transition function $\pphi$
is identified with a $k$-uple of maps (or, for short, a
\emph{$k$-map}) $\pphi_{\a} : Q \to Q$, as $\pphi_{\a}(p) \equiv
\Pphi(p,\a)$ (in this case, to avoid confusion, we use $\Pphi$ for the
$k$-uple of $\{\pphi_{\a} \}_{1 \leq \a \leq k}$).  And, clearly, a
$k$-map is identified with the corresponding vertex-labeled,
edge-coloured digraph over $n$ vertices, with uniform out-degree $k$,
such that, for each vertex $i \in [n]$ and each colour $\a \in [k]$,
there exists exactly one edge of colour $\a$ outgoing from $i$. A
terminology of graph theory will occasionally beused in the following.

We use the word \emph{motif} for an unlabeled oriented graph $M$, when
it is intended as denoting the class of subgraphs of a $k$-map
that are isomorphic to $M$. The core of our proof is in the analysis
of the probability of occurrence of certain motifs, that we now
introduce.
\begin{definition}
A \emph{M-motif} $M$ of a transition structure $\Pphi$ is a pair of states
$i \neq j$, and an ordered $k$-uple of states 
$\{\ell_{\a}\}_{1 \leq \a \leq k}$, such that
$\pphi_{\a}(i)=\pphi_{\a}(j)=\ell_{\a}$ (see Figure \ref{fig.motif},
left). Repetitions among $\ell_{\a}$'s are allowed.

A \emph{three-state M-motif} $M^{(3)}$ of a transition structure
$\Pphi$ is the analogue of a M-motif, with three distinct states $i$,
$j$ and $h$, such that
$\pphi_{\a}(i)=\pphi_{\a}(j)=\pphi_{\a}(h)=\ell_{\a}$ for all $1 \leq \a
\leq k$ (see Figure \ref{fig.motif}, right).
\end{definition}

\begin{figure}[tb]
\begin{center}
\setlength{\unitlength}{43.75pt}
\begin{picture}(6,1.88)(0.27,-0.12)
\put(0.27,-0.12){\includegraphics[scale=1.25]{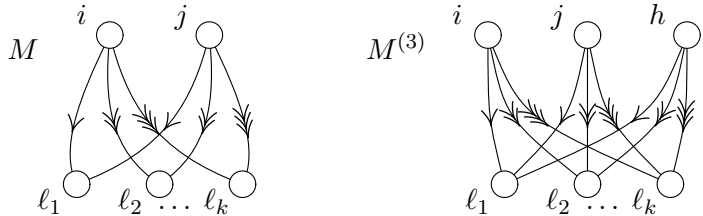}}
\put(0.1,1.25){$M$}
\put(3.2,1.25){$M^{(3)}$}
\put(0.7,1.55){$i$}
\put(1.55,1.55){$j$}
\put(0.36,-0.05){$\ell_1$}
\put(1.06,-0.05){$\ell_2$}
\put(1.4,-0.05){$\ldots$}
\put(1.8,-0.05){$\ell_k$}
\put(3.95,1.55){$i$}
\put(4.8,1.55){$j$}
\put(5.65,1.55){$h$}
\put(4.06,-0.05){$\ell_1$}
\put(4.76,-0.05){$\ell_2$}
\put(5.1,-0.05){$\ldots$}
\put(5.5,-0.05){$\ell_k$}
\end{picture}
\end{center}
\caption{\label{fig.motif}Left: the M-motif. Right: the three-state
  M-motif. The examples are for $k=3$.}
\end{figure}

\noindent
The reason for studying M-motifs is in the two following easy remarks:
\begin{remark}
If the transition structure of an automaton $A$ contains a M-motif,
with states $i$, $j$ and $\{\ell_{\a}\}$, and $\cT(i) = \cT(j)$,
then $i \sim j$ and $A$ is not minimal.
\end{remark}

\noindent

\begin{remark}
Consider a transition structure $\Pphi$ containing 
no three-state M-motifs, and
$r$ M-motifs with states 
$\big\{ i^{a}, j^{a}, \{\ell_{\a}^{a}\} \big\}_{1 \leq a \leq r}$.
Averaging over the possible sets of terminal states with the measure
$\mu_b(\cT)$,
the probability that $\cT(i^a) = \cT(j^a)$ for some 
$1 \leq a \leq r$ is $1 - (2b(1-b))^r$.
\end{remark}



\noindent
Our theorem results as a consequence of a number of statistical facts,
on the structure of random automata, which are easy to believe
although hard to prove. Thus, there is a short, non-rigorous path
leading to the theorem, that we now explain.
\begin{enumerate}
\item A fraction $1-o(1)$ of non-minimal automata contains two
  Myhill-Nerode-equivalent states
  $i \sim j$, which are the incoming states of a M-motif.
\item Random transition structures locally ``look like'' random
  $k$-maps -- this despite the highly non-local, and non-trivial,
  accessibility condition -- the only remarkable difference being in
  the distribution of the incoming degrees $r$ of the states, $p_r=0$
  if $r=0$, and $\frac{1}{\w_k} \poiss_{k \w_k}(r)$ if $r \geq 1$.
\item With this in mind, it is easy to calculate that the average
  number of M-motifs with equivalent incoming states is
  $(1-2b(1-b)) \binom{n}{2} n^{-k} 
  \left[ \frac{\eee{(r(r-1) p_r)}}{k^2} \right]^k$,
  at leading order in $n$,
  that is, $\frac{1}{2} (1-2b(1-b))\, {\w_k}^k\, n^{-k+2}$.
\item Random transition structures also show weak correlations between
  distant parts, and M-motifs are `small', thus, with high
  probability, pairs of M-motifs are non-overlapping. This suggests
  that the distribution of the number of M-motifs is a Poissonian,
  with the average calculated above (as if the corresponding events
  were decorrelated). As a corollary, we get the
  probability that there are no M-motifs. By the first claim, on the
  dominant role of M-motifs, this allows to conclude.
\end{enumerate}



\section{Structure of the proof}
\label{sec.proofstruct}

\noindent
As it often happens,
what seems the easiest way to
get convinced of a claim is not necessarily the easiest path to
produce a rigorous proof. Our proof strategy will be in fact very
different from the sequence of claims collected above. As it is quite
composite, in this section we will outline the decomposition of the
proof into lemmas, and postpone the proofs to Section
\ref{sec.proofs}.

Call $P_{\rm rare}$ the probability, w.r.t.\ $\mu_b(\Pphi,\cT)$ above,
that the transition structure contains no M-motifs,
and still the automaton is non-minimal.  Call
$P_{\rm confl}$ the probability that the transition structure contains
some three-state M-motif. Call $P(r)$ the probability that the
transition structure contains no three-state M-motifs, and exactly $r$
M-motifs.  Thus $1=P_{\rm confl} + \sum_{r \geq 0} P(r)$.

The fraction of pairs $(\Pphi,\cT)$, of transition structures $\Pphi$
with no three-state M-motifs, and lists of terminal states $\cT$ taken with
the Bernoulli measure of parameter $b$,
such that $\cT(i^a) = \cT(j^a)$ for some M-motif, is
$\sum_r
P(r)
\left(
1 - (2b(1-b))^r
\right)
$.
As a consequence, w.r.t.\ the measure $\mu_b(A)$ above, the probability
that an automaton is non-minimal is
\be
\begin{split}
\proba(\textrm{$A$ is non-minimal})
&=
\sum_r
P(r)
\left(
1 - (2b(1-b))^r
\right)
+
\mathcal{O}(P_{\rm rare})
+
\mathcal{O}(P_{\rm confl})
\ef.
\end{split}
\end{equation}
If one can prove that $P_{\rm rare}, P_{\rm confl} =
o\big( 1-P(0) \big)$,
then
\be
\begin{split}
\proba(\textrm{$A$ is non-minimal})
&=
\sum_{r \geq 1}
P(r)
\big(
1 - (2b(1-b))^r
+o(1)
\big)
\ef.
\end{split}
\end{equation}
In particular, if we can prove that $P(r) = \poiss_{\rho}(r)
(1+o(1))$, with $\rho = \sum_r r P(r)$, it would follow that
\be
\begin{split}
\proba(\textrm{$A$ is non-minimal})
&=
\big(1 - e^{-\rho (1-2b(1-b))}\big)
(1+o(1))
\ef.
\end{split}
\end{equation}
This corresponds to the statement of Theorem~\ref{theo.fullshort},
with $\rho=c_k n^{-k+2}$.

Note that our error term is not only small w.r.t.~$1$, but also, as
important for probabilities, it is small also w.r.t.\ $\min(p,1-p)$,
with $p$ the probability of our event of interest.  As, for an
alphabet with $k$ letters, $p \sim n^{-k+2}$ has a non-trivial scaling
with size when $k>2$, this difference is relevant.

So we see that Theorem \ref{theo.fullshort} is implied by
\begin{proposition}
\label{prop.enums}
The statements in the following list do hold
\begin{enumerate}
\item $P(r) = \poiss_{\rho}(r)(1+o(1))$, for some $\rho$;
\item $\rho = c_k n^{-k+2} (1+o(1))$;
\item $P_{\rm confl} = o(n^{-k+2})$;
\item $P_{\rm rare} = o(n^{-k+2})$.
\end{enumerate}
\end{proposition}
This is the theorem we will ultimately prove.

A collection of related, more explicit probabilistic statements is the
following
\begin{proposition}
\label{prop.motifaver}
For M-motifs $M$, and three-state $M$-motifs $M^{(3)}$, the average
number of occurrences in uniform random transition structures is
given by
\begin{align}
\label{eq.543254325a}
\bbm[M]
&=
\frac{1}{2}
n^{-k+2}
{\w_k}^k
\; \big( 1+o(1) \big)
\ef;
\\
\bbm[M^{(3)}]
&=
\frac{1}{6}
n^{-2k+3}
{\w_k}^{2k}
\; \big( 1+o(1) \big)
\ef.
\label{eq.543254325b}
\end{align}
Given that there are no three-state M-motifs, 
the average number of $r$-uples $(M_1,\ldots,M_r)$ of distinct
M-motifs is given by
\be
\frac{1}{r!}
\bbm \big[(M_1,\ldots,M_r) \big]
=
\frac{1}{r!}
\left(
\frac{1}{2}
n^{-k+2}
{\w_k}^k
\; \big( 1+o(1) \big)
\right)^r
\ef.
\label{eq.5544765}
\end{equation}
\end{proposition}
The proof of this proposition is postponed to Section \ref{sec.withProofPrare}.

Equation (\ref{eq.543254325a}) proves
$\rho = c_k n^{-k+2} \, \big( 1+o(1) \big)$, that is, Part~2 of
Proposition~\ref{prop.enums}.  Using the first-moment bound, 
equation (\ref{eq.543254325b}) proves
$P_{\rm confl}=\mathcal{O}(n^{-k+1})$ as required for Part~3 of
Proposition~\ref{prop.enums}.

The result in (\ref{eq.5544765}) concerning higher moments of M-motifs
implies the proof of convergence of $P(r)$ to a Poissonian, Part~1 of
Proposition~\ref{prop.enums}.  The idea behind this claim is the fact
that the occurrence of a M-motif with given states $\{i,j\}$ (and any
$k$-uple $\{\ell_{\a}\}$) is a `rare' event, as it has a probability
$\sim n^{-k}$, and, as the motifs are `small' subgraphs, involving
$\mathcal{O}(1)$ vertices, and parts of the transition structure
$\Pphi$ far away from each other (in the sense of distance on the
graph) are weakly correlated, we expect the ``Poisson Paradigm'' to
apply in this case, as discussed, for example, in Alon and Spencer
\cite[ch.\;8]{AlonSpe}.
A rigorous proof of this phenomenon can be
achieved using the strategy called \emph{Brun's sieve} (see
e.g.\ \cite[sec.\;8.3]{AlonSpe}). The verification of the hypotheses
discussed in the mentioned reference is exactly the statement of
equation~(\ref{eq.5544765}).

Thus, assuming Proposition \ref{prop.motifaver}, there is a single
missing item in our `checklist', namely, Part~4 of
Proposition~\ref{prop.enums}. We need to determine that
$P_{\rm rare}=o(n^{-k+2})$. The idea behind this is that, in absence
of M-motifs, with probability $1-o(n^{-k+2})$, for all pairs of states
$(i,j)$, the simultaneous breadth-first search trees started from $i$
and $j$ visit almost surely a large number of distinct states (for
our proof, it would suffice $\sim -\frac{\ln{n}}{\ln(1-2b(1-b))}$, but
it will turn out to be provably at least $\sim n^{\frac{1}{4(k+1)}}$
and in fact conjecturally $\mathcal{O}(n)$). Thus, as, for all the
pairs of homologous but distinct states, the states need to be either
both or none terminal states, this produces a factor $1-2b(1-b)$ per
pair.


Note that we need only an upper bound on $P_{\rm rare}$ (and no lower
bound), and we have some freedom in producing bounds, as, at a
heuristic level, we expect $P_{\rm rare} = \mathcal{O}(n^{-k+1}) \ll
o(n^{-k+2})$.  Our proof strategy will exploit this fact, and the
following property of accessible transition functions (see
\cite{cn12}): given a random $k$-map 
$\Pphi=\{\pphi_{\a}(i)\}_{1 \leq i \leq n, 1 \leq \a \leq k}$, the
number of states accessible from state $1$ is a random variable
$m=m(n,k)$, with average $\Theta(n)$ and probability around the modal
value\footnote{I.e., the most probable value.} of order
$n^{-\frac{1}{2}}$. Remarkably, given that the accessible part has
size $m$, then the induced transition structure is sampled uniformly
among all transition structures of size $m$.

This has a direct simple consequence: if the average number of
occurrences of a family of events on a random $k$-map is
$\bbm[\{\cE_i\}]_{\textrm{$k$-maps}} = \mathcal{O}(n^{-\gamma})$, then the
same average over random accessible transition functions of fixed
size is bounded as
$\bbm[\{\cE_i\}]_{\rm acc.} \leq \mathcal{O}(n^{-\gamma+\frac{1}{2}})$.
Actually, this bound is very generous and, if needed (but this is not
our case), the extra exponent $\frac{1}{2}$ could be dumped
significatively with some extra effort.

Thus, instead of proving that $P_{\rm rare} = o(n^{-k+2})$, we will
define the quantity $P'_{\rm rare}$, exactly as $P_{\rm rare}$ but on
random $k$-maps over $n$ states. Note that the definition of $P_{\rm rare}$ and
$P'_{\rm rare}$ is based on two notion: not containing certain motifs,
and not presenting pairs of Myhill-Nerode-equivalent states,
and that both this notions are not confined to accessible
automata, but are well-defined also for maps which are not accessible.
Then we will prove that 
\begin{proposition}
\label{prop.afterp5p4}
$P'_{\rm rare} =
o(n^{-k+\frac{3}{2}})$.
\end{proposition}
In summary, as this proposition implies Part~4 of
Proposition~\ref{prop.enums}, Proposition \ref{prop.motifaver} implies
Parts 1 to 3 of Proposition~\ref{prop.enums}, and
Proposition~\ref{prop.enums} implies our main Theorem
\ref{theo.fullshort}, providing proofs of 
Propositions \ref{prop.motifaver} and \ref{prop.afterp5p4} is
sufficient at our purposes. This task is fulfilled in the following
sections.


\section{Proofs of the lemmas}
\label{sec.proofs}

\noindent
{\it Proof of Proposition~\ref{prop.afterp5p4}.}
In a $k$-map, we say that a state $i$ is a \emph{sink state} if
$\pphi_{\a}(i)=i$ for all $\a$. We say that two states $\{i,j\}$ form a
\emph{sink pair} if the set
\[
N_{ij} = \{i,j,\pphi_1(i),\pphi_1(j),\cdots,\pphi_k(i),\pphi_k(j)\}
\]
has cardinality $k+1$ or smaller. As easily seen through first-moment
bound, the probability of having any sink state or sink pair in a
random $k$-map is at most of order $n^{-k+1}$ (precisely, the overall
constant is bounded by $1+\frac{(k+1)^{2k}}{2(k-1)!}$). So, at the aim
of proving that $P'_{\rm rare} = o(n^{-k+\frac{3}{2}})$, we can
equivalently conditionate the $k$-map not to contain any sink motif.

We say that two states $\{i,j\}$ form a \emph{quasi-sink pair} if the
set $N_{ij}$
has cardinality $k+2$. The average number of quasi-sink pairs in a
random $k$-map is of order $n^{-k+2}$, thus this case must be analysed
at our level of accuracy.

There exist three families of quasi-sink pairs:
those producing a M-motif, those such that there exists a value $\a$
such that $\{i,j,\pphi_{\a}(i),\pphi_{\a}(j)\}$ are all distinct
(type-1), and those such that for $h$ letters of the alphabet
$\pphi_{\a}(i)$ is uniquely realized in $N_{ij}$, and for the remaining
$k-h$ letters $\pphi_{\a}(j)$ is uniquely realized in $N_{ij}$
(type-2). 
In evaluating $P'_{\rm rare}$, we have excluded the M-motif
case, and we are left only with type-1 and type-2
quasi-sinks. Furthermore, we have excluded sink states, so in type-2
quasi-sinks we must have both $h$ and $k-h$ non-zero.

For a type-1 quasi-sink $\{i,j\}$, define the 
\emph{pair following $\{i,j\}$} as the pair $\{i',j'\}$ such that
$i'=\pphi_{\a}(i)$, $j'=\pphi_{\a}(j)$, for $\a$ the first lexicographic
letter such that $\{i,j,\pphi_{\a}(i),\pphi_{\a}(j)\}$ are all distinct.
For a type-2 quasi-sink $\{i,j\}$
define the \emph{pair following $\{i,j\}$} as the pair $\{i',j'\}$
with $i'=\pphi_{1}(i)$, $j'=\pphi_{1}(j)$.  Again, by first-moment
estimate, the probability that there exists a quasi-sink pair
$\{i,j\}$, such that also the pair following it is a quasi-sink, is
bounded by $\mathcal{O}(n^{-k+1})$ (use at this aim that $h(k-h)>0$ in
a type-2 quasi-sink), and we can conditionate our $k$-map not to
contain such motifs. If $\{i,j\}$ is a quasi-sink pair, a necessary
condition for $i \sim j$ is that also $i' \sim j'$. Thus, we can bound
$P'_{\rm rare}$ by the probability that there exist no non--quasi-sink
pairs in the $k$-map. This is the formulation of the problem that we
ultimately address.

Consider a non--quasi-sink pair $\{i,j\}$, and construct the
lexicographic breadth-first tree exploration, simultaneously on the
two states $i$ and $j$, neglecting those branches in which, in one or
both of the two trees, there is a state already visited by the
exploration (call \emph{leaves} these nodes).

Call $(v_1, v_2, \ldots)$ the ordered sequence of steps in the
breadth-first search, at which a leaf node is visited.  For fixed
values $v$ and $h$, we want to determine the probability of the event
$v_{h} \leq v$, conditioned to the event that the list has at least
$h$ items.  By standard estimate of factorials, and crucially making
use of the exclusion of sink and quasi-sink motifs, it can be proved
for this quantity
\be
\begin{split}
\proba(v_h \leq v)
&
\leq
\frac{1}{h!}
\left(
\frac{v(v+1)}{n-2v}
\right)^{h}
\ef.
\end{split}
\label{eq.detached28}
\end{equation}
Set now $h=k+1$. By definition, in a non--quasi-sink pair, we
certainly have at least $k+1$ entries $v_j$.  If
$v=\mathcal{O}(n^{\gamma})$ for some $0<\gamma<1$, we have that for
each non--quasi-sink pair $\{i,j\}$
\be
\proba(v_{k+1}^{(ij)} \leq v)
\leq
\mathcal{O}(n^{-(k+1)(1-2\gamma)})
\ef.
\end{equation}
The number of non--quasi-sink pairs is bounded by $\binom{n}{2}$, thus
by first-moment bound
\be
\proba(v_{k+1}^{(ij)} \leq v \textrm{~for all $\{i,j\}$})
\leq
\mathcal{O}(n^{-k+1+2\gamma(k+1)})
\ef.
\end{equation}
For $\gamma < \frac{1}{4(k+1)}$ we thus get 
$ \proba(v_{k+1}^{(ij)} \leq v \textrm{~for all $\{i,j\}$}) \leq
o(n^{-k+\frac{3}{2}}) $
as needed. Thus, we know that, with probability larger than
$1-o(n^{-k+\frac{3}{2}})$, all the non--quasi-sink pairs in our
$k$-map have $v_{k+1} \gtrsim n^{\gamma}$, for any $\gamma <
\frac{1}{4(k+1)}$. This means that, if we truncate the breadth-first
search tree exploration to a depth $\sim \gamma \frac{\ln n}{\ln k}$,
we have at most $k$ leaves in the tree. Thus, for all the trees, we
have at least $\sim n^{\gamma}$ internal nodes, i.e.\ pairs of states
$(i',j')=(\pphi(i,u),\pphi(j,u))$, 
for which it is required
$\cT(i')=\cT(j')$ for $i \sim j$. But, as all these states appear not
repeated in the exploration, the probability that $i \sim j$ is
bounded by an exponential of the form $(1-2b(1-b))^{n^{\gamma}}$,
which decreases faster than any power law. The overall factor
$\binom{n}{2}$ from the first-moment bound is irrelevant, and we are able
to conclude that $P'_{\rm rare} = o(n^{-k+\frac{3}{2}})$, as needed.
Note that this proof works not only for finite values of $b$ in the
open interval $]0,1[$ (as required for our purposes), but even up to 
$b \sim n^{-\gamma}$. 
\qed


\label{sec.withTabDef}

Before passing to the proof of Proposition \ref{prop.motifaver}, we
need to recall the relation between accessible deterministic complete
automata and combinatorial objects known as \emph{$k$-Dyck tableaux}
\cite{BN07}, and determine a collection of statistical properties of
these tableaux.


Given the integers $M$ and $n$,
a \emph{tableau} $T$ in the set $\cT[M \times n]$ is a map
from $[M]$ to $[n]$
such that:
\begin{enumerate}
\item
every value $y \in [n]$ has at least one preimage;
\item
calling $x_T(y)$ the smallest preimage, 
we have $x_T(1) < x_T(2) < \ldots < x_T(n)$.  
\end{enumerate}
The tableau $T$ may be
represented graphically, on a $M \times n$ grid, by marking the $M$
pairs $\{ (x,T(x)) \}_{1 \leq x \leq M}$.  Then the conditions above
translate as follows. There is exactly one marked entry per
column. Mark in red the pairs $(x_T(y),y)$, and in black the remaining
ones: there is exactly one red entry per row, which is at the left of
all black entries in the same row (if any), and the polygonal line
connecting the red entries in sequence is monotonically increasing. We
call the collections of positions of red and black marks respectively
the \emph{backbone} $B_T$ and \emph{wiring part} $W_T$ of the tableau
$T$. It is easily seen that the number of tableaux in 
$\cT[M \times n]$ is given by the 
\emph{Stirling number of second type}, $\left\{ \atopp{M}{n}
\right\}$, i.e., the number of ways of partitioning $M$ elements into
$n$ non-empty blocks (see e.g.\ \cite[sec.\;6.1, 7.4]{GKP}). The
asymptotic evaluation of $\left\{ \atopp{M}{n} \right\}$, for $n$
large and $M/n = \mathcal{O}(1)$, can be done through the general
methods of analytic combinatorics (see e.g.\ \cite{flaj}, and in
particular \cite{good} for this specific problem). A result of this
calculation that we shall need is the following

\begin{proposition}
\label{prop.stirlfugac}
If $M, M' = \kappa n + \mathcal{O}(1)$, with $\kappa>1$, calling $\w$
the only solution of the equation $-\kappa \w = \ln (1-\w)$ in $[0,1]$,
\be
\left\{ \atopp{M}{n} \right\}
=
\left\{ \atopp{M'}{n} \right\}
\left(
\frac{n}{\w}
\right)^{M-M'}
\; \big( 1+o(1) \big)
\ef.
\end{equation}
\end{proposition}

\begin{figure}[tb]
\[
\includegraphics[scale=0.72]{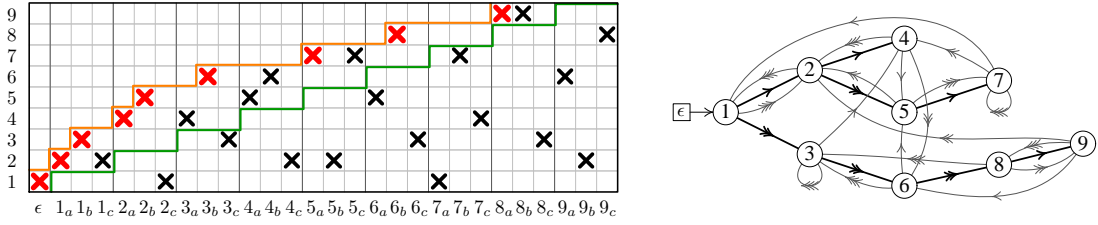}
\]
\caption{\label{fig.tabex}Left: a tableau with $n=9$ and $k=3$. The
  backbone part is in red.  This tableau is valid because the red
  entries are monotonic (as shown by the orange profile), and $k$-Dyck
  because they are all on the left of the green staircase line. Right:
  the associated $k$-map. Backbone edges, corresponding to the
  breadth-first search tree, are thick, and wiring edges are in gray.}
\end{figure}

\noindent
For a fixed integer $k$, when $M=N(n,k)=kn+1$, we have a special
subfamily of tableaux in $\cT[N \times n]$.  A tableau is
\emph{$k$-Dyck} if $x_T(\ell) \leq k (\ell-1) + 1$, i.e.\ if the
backbone cells lie above the line of slope $1/k$ containing the origin
of the grid. A small example of $k$-Dyck tableau is shown in Figure
\ref{fig.tabex}.

There exists a canonical bijection between $k$-Dyck tableaux and
transition structures $\Pphi$ of accessible deterministic complete
automata. It suffices to associate the indices $(1,2,\ldots,n)$ of the
states to the rows of the tableaux, and the indices $(\epsilon,
1_1,\ldots,1_k,\cdots,$ $n_1,\ldots,n_k)$ of the oriented edges of $\Pphi$
to the columns. Then, for $x=i_{\a}$, the entry $(x,y)$ is marked in
$T$ if and only if $\pphi_{\a}(i)=y$, and it is part of the backbone if
and only if it is part of the breadth-first search tree on $\Pphi$
started at the initial state.

Given a function $\hat{f}(y): [n] \to [M]$, consider the restriction
of the set $\cT[M \times n]$ to tableaux $T$ in which the
backbone function $x_T(y)$ is dominated by $\hat{f}$, i.e., such that
$x_T(y) \leq \hat{f}(y)$ for all $1\leq y \leq n$.  Call
$\cT[M \times n; \hat{f}]$ this set.
Our $k$-Dyck tableaux correspond to the special case
$\cT[N \times n; \hat{f}^{\emptyset}]$, with 
$\hat{f}^{\emptyset}(y) := N-k(n-y+1)$.  A required technical lemma,
that we state without proof, is the following

\begin{proposition}
\label{prop.quasidyck}
Take an integer $n$, $N=\mathcal{O}(n)$, $B=\mathcal{O}(1)$, and $\ell
\gg \sqrt{n}$.
Let $M = N - B$, and take a function $\hat{f}$ such that
$\hat{f}(y) = \hat{f}^{\emptyset}(y)$
for all $y \leq \ell$,
$\hat{f}(y) = \hat{f}^{\emptyset}(y)-B$
for all $y \geq n-\ell$,
and 
$\hat{f}^{\emptyset}(y)-B \leq \hat{f}(y) \leq \hat{f}^{\emptyset}(y)$
for all $y$. Then
\be
\frac{
\big| cT[M \times n; \hat{f}] \big|
}
{\big| cT[M \times n] \big|}
-
\frac{
\big| cT[N \times n; \hat{f}^{\emptyset}] \big|
}
{\big| cT[N \times n] \big|}
=
o(1)
\ef.
\end{equation}
\end{proposition}
\noindent
With these tools at hand, we are now ready to prove Proposition
\ref{prop.motifaver}.

\label{sec.withProofPrare}

\bigskip
\noindent
{\it Proof of Proposition \ref{prop.motifaver}.}
Given three distinct states $i$, $j$, $h$, with $i<j<h$, call
$\cM_{ijh}(T)$ the event that in the tableau $T$ there is a
three-state motif on states $\{i,j,h\}$ and $\{\ell_{\a}\}$, for some
$\ell_{\a}$'s.  Similarly, given $2r$ distinct states
$\{(i_a,j_a)\}_{1 \leq a \leq r}$, with $i_a<j_a$ and $j_a<j_{a+1}$,
call $\cM_{(i_1,j_1;\ldots;i_r,j_r)}(T)$ the event that in the tableau
$T$ there is a $r$-uple of $M$-motifs, such that the $a$-th motif has
states $i_a$, $j_a$, and $\{\ell_{\a}^a\}$, for some $\ell_{\a}^a$'s.
Proposition \ref{prop.motifaver} consists in evaluating the two
quantities
\begin{align}
&
\sum_{i<j<h} \eee{} \iver{\cM_{ijh}}_{\cT[N \times n; \hat{f}^{\emptyset}]}
\ef;
&
&
\sum_{(i_1,j_1;\ldots;i_r,j_r)}
\eee{}
\iver{\cM_{(i_1,j_1;\ldots;i_r,j_r)}
}_{\cT[N \times n;\hat{f}^{\emptyset}]}
\ef.
\label{eq.765345437}
\end{align}
We now make a crucial remark: given a backbone structure $B$, the
average over all possible completions of the indicator variables
$\iver{\cM_{ijh}}$ (respectively $\iver{\cM_{(i_1,j_1;\ldots;i_r,j_r)}}$) is zero if
any column of index in the set
$C =\{k(j-1)+1+\alpha, k(h-1)+1+\alpha \}_{1 \leq \alpha \leq k}$
has a red mark
(respectively, in the set
$C = \{k(j_a-1)+1+\alpha \}_{1 \leq a \leq r; 1 \leq \alpha \leq k}$),
otherwise, it is $\prod_{i\in C} y_i^{-1}$, where $y_i$ is the height
of the backbone profile at column~$i$.
As a consequence, backbone structures contributing to the quantities
in (\ref{eq.765345437}), weighted with the factor $\mu(\bc) \prod_{i
  \in C} y_i^{-1}$, correspond to generic backbone structures,
weighted with the factor $\mu(\bc)$, over $(N-kr) \times n$
tableaux. The correspondence is done by just erasing the columns in
$C$. The function $\hat{f}$ is modified accordingly.
Define
\be
\hat{f}^{i_1,\ldots,i_r}
(y)
=\hat{f}^{\emptyset}(y)
- k \sum_{a=1}^r
\iver{y \geq j_a}
\ef.
\end{equation}
Then, the precise statement of the remark above is
\begin{align}
\eee{} \iver{\cM_{ijh}}_{\cT[N \times n; \hat{f}^{\emptyset}]}
&=
\frac{
\big|
\cT[(N-2k) \times n; \hat{f}^{j,h}]
\big|
}{
\big|
\cT[N \times n; \hat{f}^{\emptyset}]
\big|
}
\ef;
\label{eq.765345437b0}
\\
\eee{} \iver{\cM_{(i_1,j_1;\ldots;i_r,j_r)}
}_{\cT[N \times n;\hat{f}^{\emptyset}]}
&=
\frac{
\big|
\cT[(N-kr) \times n;\hat{f}^{j_1,\ldots,j_r}]
\big|
}{
\big|
\cT[N \times n; \hat{f}^{\emptyset}]
\big|
}
\ef.
\label{eq.765345437b}
\end{align}
Thus, the right-hand side of (\ref{eq.765345437b0}) is just the
special case $r=2$ of (\ref{eq.765345437b}). Of course we have
\be
\frac{
\big|
\cT[(N-kr) \times n;\hat{f}^{j_1,\ldots,j_r}]
\big|
}{
\big|
\cT[N \times n; \hat{f}^{\emptyset}]
\big|
}
=
\frac{
\frac{
\big|
\cT[(N-kr) \times n;\hat{f}^{j_1,\ldots,j_r}]
\big|
}{
\big|
\cT[(N-kr) \times n]
\big|
}
}{
\frac{
\big|
\cT[N \times n; \hat{f}^{\emptyset}]
\big|
}{
\big|
\cT[N \times n]
\big|
}
}
\;
\frac{
\big|
\cT[(N-kr) \times n]
\big|
}{
\big|
\cT[N \times n]
\big|
}
\ef.
\end{equation}
We can apply Proposition \ref{prop.stirlfugac} to the rightmost
ratio. Then, if the $j_a$'s are within the range for application of
Proposition \ref{prop.quasidyck}, we can also simplify the leftmost
ratio, to get
%
\begin{align}
\eee{} \iver{\cM_{ijh}}_{\cT[N \times n; \hat{f}^{\emptyset}]}
&
\simeq
\left(
\frac{\w_k}{n}
\right)^{2k}
\ef;
\\
\eee{} \iver{\cM_{(i_1,j_1;\ldots;i_r,j_r)}
}_{\cT[N \times n;\hat{f}^{\emptyset}]}
&
\simeq
\left(
\frac{\w_k}{n}
\right)^{kr}
\ef.
\label{eq.765345437c}
\end{align}
As in Proposition \ref{prop.quasidyck} we just asked for $\ell \gg
\sqrt{n}$, which is compatible with $\ell \ll n$, the fraction of
$2r$-uples $(i_1,j_1;\ldots;i_r,j_r)$ such that some $j_a$'s are out
of range
is subleading, and, using the reasonings at the beginning of Section
\ref{sec.proofstruct}, the corresponding contribution can be included in 
$P_{\rm confl}$.


Then, the straightforward calculation of the number of triplets
$(i,j,h)$, and $2r$-uplets $\{(i_a,j_a)\}_{1 \leq a \leq r}$, at
leading order in $n$, allows to conclude.
\qed

\section{Algorithmic consequences}
\label{sec:other}

The results obtained in this paper open new possibilities for the
study in average of the properties of regular languages, and of the
average-case complexity of algorithms applied to minimal automata.
In this section we mention just a few among these consequences.
 
\begin{cor}
  Minimal automata with $n$ states over a $k$-letter alphabet can
  be randomly generated with $\O(n^{3/2})$ average complexity, using
  Boltzmann samplers.
\end{cor}

The random generator for complete deterministic accessible automata
given in \cite{BN07} is based on a Boltzmann sampler \cite{boltz}, its
average complexity is $\O(n^{3/2})$. As from
Theorem~\ref{theo.fullshort} there is a constant proportion of minimal
automata amongst accessible ones, the rejection method can be
efficiently applied to randomly generate a minimal automaton.  Note
that such a generator\footnote{Available at {\tt
    http://regal.univ-mlv.fr/}} has already been implemented
in~\cite{BDN07}, though there were no theoretical result on the
efficiency of this algorithm at that time.

\begin{cor}
\label{thm:moore}
  For the uniform distribution on complete deterministic accessible
  automata, the average complexity of Moore's state minimization
  algorithm is $\Theta( n \log \log n)$.
\end{cor}

\begin{proof}
The average complexity of Moore's state minimization algorithm for the
uniform distribution on n-state deterministic automata over a finite
alphabet is $\O( n \log \log n)$ \cite{david11}. The upper bound for
accessible automata is then obtained studying the size of the
accessible part of a $k$-random map \cite{cn12, kor78}. Moreover from
\cite{BDN2009} the lower bound of Moore's algorithm applied on minimal
automata with $n$ states is $\Omega(n \log \log n)$.  Using
Theorem~\ref{theo.fullshort}, this is also a lower bound for complete
deterministic accessible automata.
\end{proof}

\begin{cor}
For the uniform distribution on complete deterministic accessible
automata, there exists a family of implementations of Hopcroft's state
minimization algorithm whose average complexity is $\O( n \log \log n)$.
\end{cor}

From~\cite{david11} a family of implementations of Hopcroft's state
minimization algorithm are always faster than Moore's algorithm. The
result follows from Corollary~\ref{thm:moore}.  In~\cite{BBC09} the
lower bound on the algorithm is proved to be $\O(n \log n)$ for
any implementation.  Though it is still unknown whether there exists
an implementation whose average complexity is $\Theta(n)$.



\newpage

\appendix

\section{Details of the proof of Proposition \ref{prop.afterp5p4}}

We perform here a detailed derivation of equation
(\ref{eq.detached28}), that has been omitted in the body of the paper.
In this section we use the notation $(n)_c \equiv
n(n-1)\cdots(n-c+1)$.

Thus we have the simultaneous breadth-first tree exploration started
at a pair $\{i,j\}$ of states, in which we do not follow the
\emph{leaf} nodes, i.e., those nodes where, in one or both of the two
trees, there is a state already visited by the exploration.

This exploration is finite (the number of steps, $L_{ij}$, being
bounded by $\sim n$), as we cannot indefinitely visit new states.
A string $\tau^{(ij)}$ in $\{0,1,2\}^{L_{ij}}$
is associated to this procedure, (just use $\tau \equiv \tau^{(ij)}$
when no confusion arises), with $\tau_s$ corresponding to the number
of states already visited, among the two involved with the $s$-th
step.

The exclusion of sink and quasi-sink motifs leads to the fact that,
among $\{\tau_1,\ldots,\tau_k\}$ there must be at least a value
zero. Say $\a^*_{ij}$ is the first lexicographc letter with this
property. A further consequence is that we have at least $k+1$
non-zero entries $\tau_s$, for $1 \leq s \leq L_{ij}$. 

Recognize that the sequence $(v_1, v_2, \ldots)$ is exactly the
ordered sequence of positions $s$ at which $\tau_s>0$.  Call
$|\tau|_s=\sum_{t\leq s} \tau_t$.

We now fix $v \ll n/2$, and $h = \mathcal{O}(1)$.
The probability for the $h$-uple $(v_1,\ldots,v_h)$ is
\begin{align}
P(v_1,\ldots,v_h)
&=
\sum_{\{\tau_{v_j}\} \in \{1,2\}^h}
\frac{1}{n^{2 v_h}}
(n)_{2v_h - |\tau|_{v_h}}
\bigg(  \prod_{j=1}^{h} W_j  \bigg)
\ef;
\\
W_j
&=
\left\{
\begin{array}{ll}
2(v_j-|\tau|_{v_j-1})+1 &  \tau_{v_j}=1 \\
(v_j-|\tau|_{v_j-1})^2  &  \tau_{v_j}=2
\end{array}
\right.
\end{align}
We can bound from above the probability that $v_h \leq v$.
\be
\begin{split}
\proba(v_h \leq v)
&=
\!\!
\sum_{\substack{
(v_1,\ldots,v_h) \\ v_h \leq v }}
\!\!\!
P(v_1,\ldots,v_h)
\leq
\!\!
\sum_{\substack{
(v_1,\ldots,v_h) \\ v_h \leq v }}
\sum_{\{\tau_{v_j}\} \in \{1,2\}^h}
\!\!\!
\frac{(n)_{2v_h}}{n^{2 v_h}}
(n-2v)^{-|\tau|_{v_h}}
\bigg(  \prod_{j=1}^{h} W_j  \bigg)
\\
&\leq
\!\!
\sum_{\substack{
(v_1,\ldots,v_h) \\ v_h \leq v }}
\sum_{\{\tau_{v_j}\} \in \{1,2\}^h}
\prod_{j=1}^{h} 
\left(\frac{2v_j+1}{n-2v}\right)^{\tau_{v_j}}
\leq
\!\!
\sum_{\substack{
(v_1,\ldots,v_h) \\ v_h \leq v }}
\!\!
\left( n-2v \right)^{-h}
\prod_{j=1}^{h} 
(2v_j+2)
\ef;
\end{split}
\end{equation}
where in the last passage we used the fact that $\frac{1}{n-2v}<1$.
Now we use the fact that, for $f(v_1,\ldots,v_h)$ a positive function,
\be
\sum_{\substack{
v_1 < \ldots <v_h \\ v_h \leq v }}
f(v_1,\ldots,v_h)
\leq
\frac{1}{h!}
\sum_{\substack{
v_1, \ldots, v_h \\ v_j \leq v }}
f(v_1,\ldots,v_h)
\end{equation}
to get
\be
\begin{split}
\proba(v_h \leq v)
&
\leq
\frac{
\left( n-2v \right)^{-h}
}{h!}
\sum_{\substack{
v_1, \ldots, v_h \\ v_j \leq v }}
\prod_{j=1}^{h} 
(2v_j+2)
=
\frac{1}{h!}
\left(
\frac{v(v+1)}{n-2v}
\right)^{h}
\ef.
\end{split}
\end{equation}
This gives the claim in equation (\ref{eq.detached28}).

\section{Some statistical properties of tableaux}

We investigate here some statistical properties of tableaux and
$k$-Dyck tableaux, concerning the limit distribution of the marks, and
its fluctuations. These results are interesting \emph{per se}, and, at
the aims ofthis paper, will be instrumental to determine ratios of
cardinalities of various sets of tableaux, that, in turns, are used in
the proof of Proposition \ref{prop.motifaver}.

Associate to the backbone part $B_T$ of a tableau $T$ the sequence $\bc
= (c_1,c_2,\ldots,c_n)$, as $c_{y} = x_T(y+1) - x_T(y)-1$ (let
conventionally $x_T(0)=0$ and $x_T(n+1) \equiv kn+2$).  The
$c_y$'s are non-negative integers, related to the incremental
steps in the tableau shape $x_T(y)$.


Call $\mu(\bc)$ the number of tableaux
having the given backbone $\bc$.
This quantity is easily determined. For $\bc \in \mathbb{N}^M$,
$\mu(\bc)=0$ if $\sum_y c_y \neq (k-1)n+1$, and otherwise
\be
\label{eq.muc}
\mu(\bc)
=
\prod_{y=1}^n y^{c_y}
\ef.
\end{equation}
%
The property of a tableau $T=(B_T,W_T)$ of being $k$-Dyck depends only
on the backbone part $B_T$, and we can naturally talk of
\emph{$k$-Dyck backbones}. The factorization above still holds for
$k$-Dyck tableaux, if it is intended that $\sum_B$ is restricted to
$k$-Dyck backbones.


The backbone profile has a definite limit shape for large $n$, that we
can readily determine. Define the function $f_{\kappa}(y):[0,1]\to[0,k]$
\be
f_{\kappa}(y) = \lim_{n \to \infty}
\frac{1}{n}
\eee{\,x_T(ny)}
\end{equation}
where the average is taken w.r.t.\ the uniform measure
on $\cT[(\lfloor \kappa n \rfloor +1) \times n]$.
Provided that we have pointwise convergence to a
differentiable function (as we will see, this is the case here), this
function is translated into marginals on the sequence $\bc$, through
$\frac{{\rm d}f_{\kappa}(y)}{{\rm d}y} 
= 1 + \lim_{n \to \infty} \eee{\,c_{yn}}$.


We deal with the overall constraint by introducing a Lagrange
multiplier, associated to the horizontal width of the tableau, (to be
tuned later on in order to have concentration on the appropriate width
value $\lfloor \kappa n \rfloor +1$). The resulting measure is
\be
\label{eq.5485242gr5}
\mu_{\omega} (\bc)
=
\prod_{y=1}^n
(\omega y)^{c_y}
\ef.
\end{equation}
Now the variables $c_y$ are independent geometric variables, with
parameter $\rho = \omega y$ (i.e., $p_y(c) = (1-\rho) \rho^c$).
Average and variance are given by
\begin{align}
\eval{c_y}_{\omega}
&=
\frac{\omega y}{1-\omega y}
\ef;
&
\label{eq.variaCy}
\eval{c_y^2}_{\omega}
-
\eval{c_y}_{\omega}^2
&=
\frac{\omega y}{(1-\omega y)^2}
\ef.
\end{align}
The value of $\omega$ is determined by the equation
$
\sum_{y=1}^n
\frac{\omega y}{1-\omega y}
=
(\kappa -1)n+1
$,
that is, in the large $n$ limit,
\be
\int_{0}^n
\dx{y}
\frac{\omega y}{1-\omega y}
=
(\kappa -1)n
+ \mathcal{O}(1)
\ef.
\end{equation}
Through the scaling $y \to y/n$, $\omega \to \omega n$ we can extract
the leading contribution
$
\int_{0}^1
\dx{y}
\frac{\omega y}{1-\omega y}
=
\kappa -1
$,
and, as we have
\be
\label{eq.87657876}
\int_{y=0}^Y
\dx{y}
\frac{\omega y}{1-\omega y}
=
-Y - \frac{\ln(1-\omega Y)}{\omega}
\ef,
\end{equation}
we get that value $\omega_\kappa $ for the multiplier is the root in the
interval $(0,1)$ of the transcendental equation
\be
\label{eq.defw1}
\kappa 
=
-\frac{\ln(1-\omega)}{\omega}
\ef,
\end{equation}
that is, the same constants defined in 
(\ref{eq.67476548}).
Then, the limit curve is just deduced from (\ref{eq.87657876}) with
$\omega=\omega_\kappa $:
\be
\label{eq.87657876b}
f_\kappa (y)
=
-\frac{\ln(1-\omega_\kappa  y)}{\omega_\kappa }
\ef.
\end{equation}
Note that the derivative of $f_\kappa (y)$ at $y=0$ and $y=1$ 
are
$
f'_\kappa (0) = 1
$
and
$
f'_\kappa (1) = \frac{1}{1-\w_\kappa }
$,
which are respectively smaller and larger than $\kappa $, for any $\kappa>1$ real.
More generally, $\big( f'_\kappa (y) \big)^{-1}=1-\w_\kappa  y$ is the density of
backbone marks around $x=f_\kappa (y)$. As every column is marked, either in
red or in black, the density of black marks around $x=f_\kappa (y)$ is 
$\w_\kappa  y$. As the position of a black mark in a given column is chosen
uniformly in the range $\{1,\ldots,y\}$, the probability of putting
a black mark in a given position $(x,y)$, provided that $x > f_\kappa(y)$,
is $\w_\kappa /n$, notably \emph{regardless of $x$ and $y$}.

A further useful property of the backbone is the calculation of the
variance, in the system with the Lagrange multiplier (and thus without
the constraint $\sum_y c_y = (k-1)n+1$), which is given by the
integral
\be
\label{eq.wievaria}
S(Y) = 
\int_{0}^Y
\dx{y}
\frac{\omega y}{(1-\omega y)^2}
=
\frac{Y}{1-\omega Y} + \frac{\ln(1-\omega Y)}{\omega}
\ef.
\end{equation}
Through the Central Limit Theorem we can deduce from this expression
the asymptotic probability for fluctuations from the limit shape. For
a given row $y$, such that $y$, $n-y \gg 1$, the probability of having
$x_T(y)=\lfloor n f_{\kappa}(y) \rfloor + \xi$ is approximatively
(using here the variance function (\ref{eq.wievaria}))
\begin{align}
p^{\rm bridge}_{n,y}(\xi)
&=
\frac{1}{\sqrt{2 \pi s}}
\exp \left[
-\frac{\xi^2}{2s}
\right]
\ef;
&
s=n
\frac{S(1)}{S(y/n) (S(1)-S(y/n))}
\ef.
\label{eq.Wie1P}
\end{align}
This is of course only the case $r=1$ of the basic formulas for the
$r$-point joint distribution in an inhomogeneous Wiener Process
$x(t)$, derived from the continuum limit of the sum of independent
random variables with variance $S(t)$, as in our case (see
e.g.~\cite[sec.\;5.6]{Kac}). However, this formula will be sufficient
at our present purposes.

We have now all the ingredients to prove
Proposition \ref{prop.quasidyck}.

\noindent
{\it Proof of Proposition \ref{prop.quasidyck}.}
We start by comparing different functions $\hat{f}$ satisfying the
constraint, at a fixed value $M$. 
Remark that any two such functions $\hat{f}_1$, $\hat{f}_2$ differ by
a number of cells bounded by $Bn$,
and that, if 
$\hat{f}_1(y) \leq \hat{f}_2(y)$ for all $y$,
$\big| cT[M \times n; \hat{f}_1] \big|
\leq \big| cT[M \times n; \hat{f}_2] \big|$.
Thus, by telescoping, up to a factor $Bn$, it suffices to estimate 
the quantity
\[
\frac{
\big| cT[M \times n; \hat{f}_2] \big|
-
\big| cT[M \times n; \hat{f}_1] \big|
}
{\big| cT[M \times n] \big|}
\]
for a pair of functions $\hat{f}_1$, $\hat{f}_2$ differing by a single
cell in the position $(x,y)$. This quantity is positive at sight.

Note that the constraint on functions $\hat{f}$ forces $y$, $n-y \gg
\sqrt{n}$. Thus, the use of (\ref{eq.Wie1P}) (based on use of the
Central Limit Theorem) is legitimate, and we have
\be
\begin{split}
\frac{
\big| cT[M \times n; \hat{f}_2] \big|
-
\big| cT[M \times n; \hat{f}_1] \big|
}
{\big| cT[M \times n] \big|}
 & \sim
p^{\rm bridge}_{n,y}(ky-n f(y)+\mathcal{O}(1))
\\ &
\sim
\exp\left(-\mathcal{O}\left( \smfrac{\min(y,n-y)^2}{n} \right)\right)
\ef,
\end{split}
\end{equation}
where the constant is positive at sight, and could be determined from
the expressions (\ref{eq.wievaria}) giving $S(y/n)$ and $S(1)-S(y/n)$
(which are of order 1), and the quantity $ky-n f_k(y)$ (with $f_k(y)$
as in (\ref{eq.87657876b})), which is of order $\min(y,n-y)$. The
precise value is accessible with some calculation, but irrelevant at
our purposes.

Now that we determined that all functions $\hat{f}$ in the appropriate
range produce the same ratio, up to an absolute error which is
exponentially small, we can evaluate
this ratio, for a reference $\hat{f}$ of our choice. We choose, for
any value $a$ such that both $a$ and $n-a$ are of order $n$,
\be
\hat{f}_a(y)
=
\left\{
\begin{array}{ll}
\hat{f}^{\emptyset}(y)    & y < n-a \\
\hat{f}^{\emptyset}(y)-B  & y \geq n-a
\end{array}
\right.
\end{equation}
For a tableau $T$, call $M'$ the value such that
$x_T(n-a)=M'<M-k(a+1)$, the latter inequality being forced
by the constraint $x_T(y)\leq \hat{f}(y)$.
We can thus express the ratio
$\big| cT[M \times n; \hat{f}] \big| / \big| cT[M \times n] \big|$
in the form
\be
\begin{split}
\frac{
\big| cT[M \times n; \hat{f}] \big|
}{
\big| cT[M \times n] \big|
}
&=
\frac{1}
{\displaystyle{
\sum_{\bc}
\mu_{\w}(\bc)
\iver{|\bc|=M-n}
}}
\sum_{M' < M-k(a+1)}
\sum_{\bc}
\mu_{\w}(\bc)
\iver{|\bc|=M-n}
\\
&
\qquad
\times
\iver{x_T(n-a)=M'}
\iver{x_T(y)\leq \hat{f}(y)}_{y>n-a}
\iver{x_T(y)\leq \hat{f}(y)}_{y<n-a}
\ef.
\end{split}
\end{equation}
The marginalisation on the value of $M'$ makes the two event
$\iver{x_T(y)\leq \hat{f}(y)}_{y>n-a}$ and
$\iver{x_T(y)\leq \hat{f}(y)}_{y<n-a}$ independent, and in fact the
second one depends only on $n-a$ and $M'$, and the
first one 
only on $n$, $a$ and $M-M'$ (not on $B$). Equivalently, as
$N=kn+1=M-B$, we can use $n$, $a$ and $M'-B$ as independent
parameters,
i.e.,
\be
\begin{split}
\frac{
\big| cT[M \times n; \hat{f}] \big|
}{
\big| cT[M \times n] \big|
}
&=
\sum_{M' < M-k(a+1)}
p(M')
\;
p^+_{n,a}(M'-B)
\;
p^-_{a}(M')
\ef.
\end{split}
\label{eq.764654765Mpre}
\end{equation}
where $p(M')$ is nothing but $p^{\rm bridge}_{n,n-a}(M'-n f(y))$, and
corresponds to the expectation of
$\iver{x_T(n-a)=M'}$ alone.
From equation (\ref{eq.Wie1P})
we know that the leading contribution to
$p(M')$ is well-approximated by a Gaussian, with mean and variance 
\begin{align}
\eee{} \left(\smfrac{M'}{M}\right)
&=
f\left(\smfrac{n-a}{n} \right)
\ef;
&
\eee{} (M')^2
-
(\eee{} M')^2
=
M
\frac{S(1)}{S(\frac{n-a}{n}) \big(S(1)-S(\frac{n-a}{n}) \big)}
\ef,
\end{align}
where, as $(M-1)/n=k+\mathcal{O}(1/n)$, up to subleading corrections
we can use the parameters $k$ and $\w_k$ in the determination of
$f(y)$ and $S(y)$. Similarly, as
the width of the Gaussian is of order $\sqrt{N}$, up to subleading
corrections we can replace $p(M')$ by $p(M'-B)$, and write, after a
translation,
\be
\begin{split}
\frac{
\big| cT[M \times n; \hat{f}] \big|
}{
\big| cT[M \times n] \big|
}
&=
\sum_{M' < N-k(a+1)}
p(M')
\;
p^+_{n,a}(M')
\;
p^-_{a}(M'+B)
\ef.
\end{split}
\label{eq.764654765M}
\end{equation}
The expression analogous to (\ref{eq.764654765M}), for $M=N$,
reads
\be
\begin{split}
\frac{
\big| cT[N \times n; \hat{f}] \big|
}{
\big| cT[N \times n] \big|
}
&=
\sum_{M' < N-k(a+1)}
p(M')
\;
p^+_{n,a}(M')
\;
p^-_{a}(M')
\ef.
\end{split}
\label{eq.764654765N}
\end{equation}
As the Korshunov constant is of order $1$ for all $k>1$, the
values of the Gaussian $p(M')$ are of order $1/\sqrt{N}$ at the
maximum, and the functions
$p^+_{n,a}(M')$ and $p^-_{a}(M')$ are (respectively decreasing and
increasing) monotonic in $M'$, we have that these functions must be of
order 1 in the region relevant for $p(M')$. Actually, in this region
we even have $p^-_{a}(M') = 1-o(1)$ (see \cite[Lemma 13]{BN07}), and
in particular, as a corollary, $p^-_{a}(M')$ is smooth possibly up to 
small corrections. Then, the comparison of (\ref{eq.764654765M}) and
(\ref{eq.764654765N}) allows to conclude.
\qed

\end{document}